\pdfoutput=1
\documentclass[11pt,letterpaper]{article}
\usepackage[margin=1in]{geometry}
\usepackage[T1]{fontenc}
\usepackage[utf8]{inputenc}
\usepackage{newtxtext}
\usepackage[hyphens]{url}
\usepackage{graphicx}
\usepackage{natbib}
\usepackage{caption}
\usepackage[ruled,vlined,linesnumbered,lined,boxed,commentsnumbered]{algorithm2e}
\usepackage{latexsym}
\usepackage{array}
\usepackage{amsmath}
\usepackage{amssymb}
\usepackage{amsfonts}
\usepackage{textcomp}
\usepackage{amsthm}
\usepackage{xcolor}

\theoremstyle{plain}
\newtheorem{theorem}{Theorem}
\newtheorem{lemma}{Lemma}

\theoremstyle{definition}
\newtheorem{definition}{Definition}
\newtheorem{example}{Example}

\newcommand{\nmodels}{\not \models}

\newcommand{\Exp}{\mathtt{Exp}}
\newcommand{\atom}{\mathtt{Atom}}

\newcommand{\formula}{\mathtt{Form}}

\newcommand{\set}[1]{\{ #1 \}}

\newcommand{\tuple}[1]{\langle {#1} \rangle}
\newcommand{\true}{\top}
\newcommand{\false}{\bot}

\newcommand{\intNum}{\mathcal{Z}}

\newcommand{\varSet}{\mathcal{V}}
\newcommand{\varSubSet}{\mathcal{U}}
\newcommand{\statevar}{\mathcal{V}}

\newcommand{\winSet}{\mathcal{W}}
\newcommand{\loseSet}{\mathcal{L}}
\newcommand{\drawSet}{\mathcal{D}}

\newcommand{\act}{\mathcal{A}}
\newcommand{\trans}{\mathcal{T}}
\newcommand{\constraint}{\mathcal{C}}
\newcommand{\estate}{\mathcal{E}}

\newcommand{\order}{\prec}

\newcommand{\suc}{\sigma}

\newcommand{\len}[1][]{|#1|}
\newcommand{\dep}[1][]{d(#1)}
\newcommand{\modulo}{\equiv}
\newcommand{\nmodulo}{\not \equiv}

\newcommand{\play}{\tau}
\newcommand\ie{{\it i.e.}}

\newcommand{\termProb}{\textsc{Term}}
\newcommand{\cyclicProb}{\textsc{Cyclic}}
\newcommand{\winFormProb}{\textsc{WinForm}}
\newcommand{\loseFormProb}{\textsc{LoseForm}}
\newcommand{\drawFormProb}{\textsc{DrawForm}}
\newcommand{\winStateProb}{\textsc{WinState}}
\newcommand{\loseStateProb}{\textsc{LoseState}}
\newcommand{\drawStateProb}{\textsc{DrawState}}

\newcommand{\haltProb}{\textsc{Halt}}
\newcommand{\mortProb}{\textsc{Mort}}

\newcommand{\periodProb}{\textsc{Perd}}

\newcommand{\var}{\mathcal{V}}

\newcommand{\assign}{\leftarrow}

\newcommand{\inc}[1][]{\mathtt{Inc}_{#1}}
\newcommand{\dec}[1][]{\mathtt{Dec}_{#1}}

\newcommand{\halt}{\mathtt{Halt}}

\newenvironment{proofsketch}{\noindent \textit{Proof sketch}:}{\qed}

\title{Decidability and Undecidability Results for \\ LIA-Definable Impartial Combinatorial Games}
\author{
    Shiguang Feng\textsuperscript{1},
    Liangda Fang\textsuperscript{2},
    Jiahao Luo\textsuperscript{2},
    Quanlong Guan\textsuperscript{2}\\[0.5ex]
    \textsuperscript{1}School of Computer Science and Engineering, Sun Yat-sen University, China\\
    \textsuperscript{2}Department of Computer Science, Jinan University, China\\[0.5ex]
    \texttt{fengshg3@mail.sysu.edu.cn},
    \texttt{fangld@jnu.edu.cn}\\
    \texttt{caroluo@stu2024.jnu.edu.cn},
    \texttt{gql@jnu.edu.cn}
}
\date{}

\begin{document}

\maketitle

\begin{abstract}
\looseness=-1
Combinatorial game theory is a branch of mathematics and theoretical computer science that typically studies deterministic games with perfect information and no elements of chance.
The majority of combinatorial games are impartial and formalized in linear integer arithmetic, which we call \textit{LIA-definable ICGs}.
In this paper, we provide some decidability and undecidability results of LIA-definable ICGs. 
Our theoretical results are:		
\begin{enumerate}
\item Deciding if an LIA-definable ICG is terminating (resp. cyclic) is undecidable, and becomes decidable for terminating LIA-definable ICGs.		


\item Deciding if an LIA formula exactly characterizes the set of winning (resp. losing/draw) states of an LIA-definable ICG is undecidable, and becomes decidable for terminating LIA-definable ICGs.

\item Deciding if a state is winning (resp. losing) even for a terminating LIA-definable ICG is undecidable, and becomes decidable for terminating finite-depth LIA-definable ICGs.

\item Deciding if a state is draw for an LIA-definable ICG is undecidable, and becomes decidable for terminating LIA-definable ICGs.
\end{enumerate}
\end{abstract}

\section{Introduction}
\looseness=-1
Combinatorial game theory is a branch of artificial intelligence and theoretical computer science that typically studies deterministic games with perfect information and no elements of chance.
A game that satisfies the following conditions is called a combinatorial game \cite{Fer2018}:
(1) There are two players and (possibly infinitely) many states; 
(2) Two players alternate moving, each transiting from one state to another; 
(3) The game starts from a \textit{legal} state and ends when it moves to a \textit{terminal state} in which the current player cannot make a move.
Under the \textit{normal play rule}, the last player to move \textit{wins}, whereas under the mis\`{e}re play rule the last player to move \textit{loses}.
If the game never ends, it is declared as a \textit{draw}.
A game is \textit{impartial} if both players has the same choice of moving; otherwise, it is called \textit{partizan}.
A game is \textit{terminating} if it always ends in a finite number of moves. 
A game is \textit{cyclic} if it contains a loop, that is, a sequence of states with repetitions. 
A game is \textit{LIA-definable} if it is formalized in linear integer arithmetic. 

\looseness=-1
The majority of combinatorial games are impartial and LIA-definable, for example, Subtraction \cite{Yag2001}, Nim \cite{Bou1901}, Wythoff \cite{Wyt1907}, Welter \cite{Wel1952} and Chomp \cite{Fre1952}.
Over the past decades, extensive efforts have been devoted to analyzing the characteristic of winning states of LIA-definable impartial combinatorial games (ICGs).
\citet{Zei2001} identified the sufficient and necessary condition of winning states in three-rowed Chomp game when the number of piles in the third row is less than $3$.
\citet{DucFNR2010} showed that several extensions and restrictions of Wythoff game preserve exactly the same set of winning states as the original game.
In addition to the above theoretical results, some automatic approaches to generating the winning formula (that is, a formula exactly capturing the set of winning states) for LIA-definable ICGs were developed \cite{WuFXLQCR2020,FangYCHGX2024}.
With only a few exceptions, no research had been carried out on the decidability issue of LIA-definable ICGs.
These exceptions include (1) deciding whether two ICGs have the same set of winning states \cite{LarW2013}; and (2) deciding whether the set of winning states for ICGs is $2$-dimensional semi-periodic \cite{Lar2013}.

\looseness=-1
In this paper, we provide some decidability and undecidability results regarding LIA-definable ICGs.
The decision problems we study in this paper include:

\begin{definition}
	\begin{itemize}
%
%
%
%

		\item The terminating (resp. cyclic) problem ($\termProb$ (resp. $\cyclicProb$)) is deciding if an LIA-definable ICG $\Pi$ is terminating (resp. cyclic)?
		
		\item The winning (resp. losing/draw) formula problem ($\winFormProb$ (resp. $\loseFormProb$/$\drawFormProb$)) is deciding if an LIA-formula $\phi$ is the winning (resp. losing/draw) formula of an LIA-definable ICG $\Pi$?
		
		\item The winning (resp. losing/draw) state problem ($\winStateProb$ (resp. $\loseStateProb$/$\drawStateProb$)) is deciding if a state $s$ is a winning (resp. losing/draw) state of an LIA-definable ICG $\Pi$?
	\end{itemize}	
\end{definition}

%
%
%
%
%
%
%
%

Our theoretical contributions are
\begin{enumerate}
	\item For general LIA-definable ICGs, all of the above decision problems are undecidable. 

	\item For terminating LIA-definable ICGs, all of the above decision problems except the winning and losing state problems become decidable.
	
	\item For terminating LIA-definable ICGs, the winning and losing state problems are still undecidable.

	\item We define a class of terminating LIA-definable ICGs, called finite-depth, and show that the winning and losing state problem are decidable for this class.
\end{enumerate}

\section{Preliminaries}
\looseness=-1

\subsection{Linear Integer Arithmetic} \label{sec:LIA}
\looseness=-1
Let $\intNum$ be the set of integers and $\varSet$ the set of variables.
\textit{Linear integer arithmetic expressions} ($\Exp$), \textit{atoms} ($\atom$) and \textit{formulas} ($\formula$) are defined by the following grammar:
\[e, e' \in \Exp :: c \mid v \mid e + e' \mid e - e' \]
\[
l \in \atom ::  e = e' \mid e \neq e' \mid e \leq e' \mid e \geq e' \mid e \modulo_{c} c' \mid e \nmodulo_{c} c'
\]
\[
\phi, \phi' \in \formula :: \false \mid \true \mid l \mid \neg \phi \mid \phi \land \phi' \mid \phi \lor \phi' \mid \forall v \phi \mid \exists v \phi
\]
where $c, c' \in \intNum$ and $v \in \varSet$.

\looseness=-1
The atom $e \modulo_{c} c'$ means that $e$ and $c'$ are congruence modulo $c$ (that is, $e - c'$ is divisible by $c$), and $e \nmodulo_{c} c'$ denotes the negation of $e \modulo_{c} c'$.

A \textit{state} $s$ is an interpretation that maps every variable $v$ to an integer $v(s)$.
For a state $s$, we can evaluate a expression $e$ into an integer $e(s)$ that the expression simplifies to when replacing each variable $v$ with its corresponding value $v(s)$.
The Boolean value $\phi(s)$ of a formula $\phi$ can be defined in a similar way.
A state $s$ \textit{satisfies} a formula $\phi$, denoted by $s \models \phi$, iff $\phi(s) = \true$.
A formula $\phi$ \textit{implies} another one $\psi$, denoted by $\phi \models \psi$, iff $s \models \psi$ for every state $s$ satisfying $\phi$.
A formula $\phi$ is \textit{valid}, iff $s \models \phi$ for every state $s$. 

\looseness=-1
It is well-known that LIA allows for quantifier elimination, that is, any LIA formula can be transformed into an equivalent quantifier-free formula \cite{Coo1972,Mon2010}.
Let $\varSubSet$ be a subset $\set{u_1, \cdots, u_n}$ of variables.
We use $\forall \varSubSet \phi$ for $\forall u_1 \cdots \forall u_n \phi$ and $\exists \varSubSet \phi$ for $\exists u_1 \cdots \exists u_n \phi$.

\subsection{LIA-Definable Impartial Combinatorial Games}
\label{sec:ICG}
\looseness=-1



\looseness=-1
\begin{definition} \label{def:game} \rm
An LIA-definable ICG is defined as a tuple $\Pi = \tuple{\var, \act, \constraint, \estate}$ where 
\begin{itemize}
\item $\var$: a finite set of variables. 
\item $\act$: a finite set of actions.			
\item $\constraint$: an LIA formula denoting all legal states.
\item $\estate$: an LIA formula denoting all terminal states. 
\end{itemize}
\end{definition}

\looseness=-1
Each \textit{action} transits predecessor states to a possibly infinite set of successor states.
To distinguish predecessor variables from successor one, we use the unprimed variable $v \in \var$ for the former and the primed variable $v' \in \var'$ for the latter.
Given a state $s$ over $\var$, we use $p(s)$ for the state over $\var'$ iff $v'(p(s)) = v(s)$ for every $v \in \var$.
Each action $a$ is defined by a transition formula $\trans(a)$ over predecessor state and successor state.
We require that for every legal state $s$, if there is a successor state $s'$ such that $s \cup p(s') \models \trans(\act)$, then $s \nmodels \estate$.
For example, the transition formula $v \geq 3 \land v' = v - 3$ means that the action is executable when $v \geq 3$ and decreases the value of $v$ by $3$.
\textit{The global transition formula} $\trans(\act)$ of $\Pi$ is defined as $\bigvee_{a \in \act} \trans(a)$.
It represents the whole state transition system of the game $\Pi$.
A state $s'$ is a \textit{successor state} of $s$, denoted by $s \order s'$, iff $s \cup p(s') \models \trans(\act)$.
We use $\sigma(s)$ for the set of successor states of $s$.
We use $\order^+$ for the transitive closure of $\order$.
A state $s'$ is a \textit{descendant state} of $s$, iff $s \order^+ s'$. 
Throughout this paper, we assume that all ICGs are LIA-definable.

A play $\play$ is a sequence of states $(s_1, s_2, \cdots)$ where $s_1$ is a legal initial state and each state $s_{i + 1}$ is a successor state of $s_i$ for $i \geq 1$.
A play $\play$ is \textit{cyclic}, iff $\play = (\cdots, s, \cdots, s, \cdots)$.
If a play $\play$ is finite, then it ends at a terminal state.
We say $\play$ \textit{starts from} $s$, iff the first state of $\play$ is $s$.
An ICG is \textit{terminating}, iff every play is finite.
An ICG is \textit{cyclic}, iff there is a cyclic play.

\looseness=-1
Any state can be classified as three categories: \textit{winning}, \textit{losing} and \textit{draw}. 
\begin{definition} \label{def:wlState} \rm
In an ICG, the sets $\winSet$ of winning states and $\loseSet$ of losing states are recursively defined as
\begin{enumerate}
	\item Every terminal state is a losing state.
	\item Every state such that some of its successor states are losing states is a winning state.
	\item Every state whose successor states are winning states is a losing state.
\end{enumerate}
\end{definition}

\looseness=-1
Definition \ref{def:wlState} characterizes the concepts of $\winSet$ and $\loseSet$ under the normal play rule.
Under the mis\`{e}re rule, the definitions of $\winSet$ and $\loseSet$ are the same as Definition \ref{def:wlState} except that every terminal state is a winning state.
For the sake of clarity, throughout this paper, we assume that ICGs follow the normal play rule.
The set $\drawSet$ of draw states can be defined as $S \setminus (\winSet \cup \loseSet)$ where $S$ is the set of legal states of an ICG.
The three sets $\winSet$, $\loseSet$ and $\drawSet$ are mutually disjoint.

\looseness=-1
To obtain the finite representation of winning (resp. losing/draw) states, we use an arithmetic formula to exactly capture this notion, called \textit{the winning (resp. losing/draw) formula}.
\begin{definition} \label{def:winForm} \rm
Let $\Pi = \tuple{\var, \act, \constraint, \estate}$ be a terminating ICG. 
We say a formula $\phi \in \formula$ is a \textit{winning (resp. losing/draw) formula} of $\Pi$, iff for every state $s$, $s \models \constraint \land \phi$ iff $s$ is a legal winning (resp. losing/draw) state.
\end{definition}

\looseness=-1
The winning formula $\phi$ along with the formula $\constraint$ exactly characterizes the set of winning states.
In general, there are many logical different winning formulas for an ICG.
But they are logically equivalent under the background theory $\constraint$, that is, $\constraint \models \phi \equiv \phi'$.
Hence, we sometimes call a winning formula the winning one.
The concepts of losing and draw formulas are similar.

%
%
%
%
%

\subsection{2-Counter Machine}

\begin{definition}[\cite{Min1961,Lam1961}] \rm \label{def:countMach}
A \textit{2-counter machine} $M$ consists of
\begin{enumerate}
\item a special index register $r_0$ which denotes the label of current instruction to be executed;

\item two count registers $r_1$ and $r_2$, each of which can hold a non-negative integer;

\item a list $L$ of labeled instructions, which comprises two types: $\inc[i, l]$, $\dec[i, l, l']$ and $\halt$ where $1 \leq i \leq 2$, $1 \leq l, l' \leq \len[L]$ and $\len[L]$ denotes the length of $L$.

\end{enumerate}
\end{definition}

\looseness=-1
Intuitively, $\inc[i, l]$ increments the value of the $i$-th count register $r_i$ and moves to the $l$-th instruction. 
The instruction $\dec[i, l, l']$ works as follows: if the value of $r_i$ is non-zero, then it decrements the value of $r_i$, and moves to the $l$-th instruction; otherwise, it directly jumps to the $l'$-th instruction.




\looseness=-1
A \textit{configuration} $C$ of a counter machine $M$ is a mapping that associate every register $r$ to a non-negative integer $C(r)$ and satisfies $1 \leq C(r_0) \leq \len[L]$.
A configuration $C$ is \textit{halting}, iff there is $C(r_0): \halt \in L$.
We use $M(C)$ to denote the single step configuration of $M$ on $C$, obtained from $C$ via executing the current instruction.
We use $M^i(C)$ to denote the $i$-th step configuration of $M$ on $C$, which is defined as $M^0(C) = C$, and $M^i(C) = M(M^{i - 1}(C))$ for $1 \leq i \leq n$.
A 2-counter machine $M$ is \textit{terminating} in a configuration $C$, iff there is a run $C, M^1(C), \cdots, M^m(C)$ s.t. $M^m(C)$ is a halting configuration.
A configuration $C$ is \textit{periodic} on $M$ iff $C = M^m(C)$ for $m > 0$. 
The following decision problems about 2-counter machine are undecidable \cite{Min1967,Blo2001,Blo2002}. 
\begin{definition}
	\begin{itemize}
%
%
%
%
		
		\item The halting problem ($\haltProb$) is deciding if a 2-counter machine $M$ is terminating in a configuration $C$?
		
		
		\item The mortality problem ($\mortProb$) is decide if a 2-counter machine $M$ is terminating on all configurations?

		\item The periodic problem ($\periodProb$) is deciding if a 2-counter machine $M$ has a periodic configuration?
	\end{itemize}
\end{definition}




\section{The Terminating and Cyclic Problems}

To prove the undecidability of the terminating and cyclic problems, we first provide a reduction from 2-counter machines $M$ to ICGs $\Pi(M)$.


By prime factorization theorem \cite{Long1972}, every positive integer $z$ can be represented uniquely as a product of prime numbers, up to the order of the factors, that is, $z = 2^i 3^j 5^k \cdots$.
Given three natural numbers $p$, $m$ and $n$ and a variable $v$ s.t. $p$ is a prime number and $m \leq n$, we define an LIA-formula $\varphi_{p, m}(v)$ as $v \modulo_{p^m} 0 \land v \nmodulo_{p^{m + 1}} 0$ and $\chi_{p, m, n}(v)$ as $v \modulo_{p^m} 0 \land v \nmodulo_{p^{n + 1}} 0$ where $v$ is a variable.
Intuitively, $\varphi_{p, m}(v)$ means that the order of the prime $p$ in the value of $v$ is $m$ while $\chi_{p, m, n}(v)$ means that the order of the prime $p$ is between $m$ and $n$.

\begin{definition} \rm \label{def:firstEncoding}
Let $M$ be a 2-counter machine and $L$ the list of labeled instructions of $M$.
An ICG $\Pi(M) = \tuple{\var, \act, \constraint, \estate}$ is defined as
\begin{itemize}
\item $\statevar: \set{v}$ and $\act: \set{a_l \mid 1 \leq l \leq \len[L] \text{ and } l: \halt \notin L}$;

\item $\constraint: v \geq 0 \land \chi_{2, 1, \len[L]}(v)$ and $\estate: \bigvee_{l: \halt \in L} \varphi_{2, l}(v)$.
\end{itemize}	

\noindent
The transition formula $\trans(a_l)$ for action $a_l$ is defined as
\begin{center}
	$\begin{cases}
		\varphi_{2, l}(v) \land 2^l v' = 2^j 3 v &\text{if } l: \inc[1, j] \in L; \\
		
		\varphi_{2, l}(v) \land 2^l v' = 2^j 5 v &\text{if } l: \inc[2, j] \in L; \\
		
		\begin{aligned}
			&\varphi_{2, l}(v)  \land \\ 
			&\quad [(\neg \varphi_{3, 0}(v) \land 2^l 3 v' = 2^j v) \lor \\ 
			&\quad\; (\varphi_{3, 0}(v) \land 2^l v' = 2^{j'} v)]
		\end{aligned}
		& \text{if } l: \dec[1, j, j'] \in L; \\
		
		\begin{aligned}
			& \varphi_{2, l}(v) \land \\
			&\quad [(\neg \varphi_{5, 0}(v) \land 2^l 5 v' = 2^j v) \lor {}\\
			&\quad\;(\varphi_{5, 0}(v) \land 2^l v' = 2^{j'} v)]
		\end{aligned}
		& \text{if } l: \dec[2, j ,j'] \in L.
	\end{cases}$
\end{center}
\end{definition}


The ICG $\Pi(M)$ uses a single variable $v$ to encode the index register and the two count registers of the 2-counter machine $M$.
A state $s$ denotes a configuration $C$, iff $v(s) = 2^{C(r_0)} 3^{C(r_1)} 5^{C(r_2)} b$, where $b$ is not a multiple of $2$, $3$, or $5$.
The constraint $\constraint$ ensures that the value of index register is in the valid range (\ie, $1 \leq C(r_0) \leq \len[L]$) in every legal state.
The termination condition $\estate$ means that the execution of instructions ends in instruction $\halt$.
The executability condition of action $a_l$ is $\varphi_{2, l}(v)$, indicating that it corresponds to the $l$-th instruction.
When the $l$-th instruction is not $\halt$, action $a_l$ emulates its behavior. 
Suppose that $l: \inc[1, j] \in L$ and $s$ is a state with the value $v(s) = 2^{l} 3^{m} 5^{n} b$.
Action $a_l$ changes the value $v(s)$ of state $s$ and obtain a successor state $s'$ s.t. $v(s') = 2^{j} 3^{m + 1} 5^{n} b$.
Since the three symbols $l$, $j$ and $j'$ are fixed integers in the termination condition $\estate$ and transition formulas $\trans(a)$ are LIA-definable.
So is $\Pi(M)$. 

\begin{theorem} \label{thm:isTerminatingICG}
$\termProb$ is undecidable.
\end{theorem}
\begin{proofsketch}
It can be verified that a 2-counter machine $M$ halts on all configurations iff the ICG $\Pi(M)$ from any state eventually ends.
Since the mortality problem of 2-counter machines is undecidable, deciding if an LIA-definable ICG is terminating is undecidable.
\end{proofsketch}

\begin{theorem} \label{thm:isCyclicICG}
$\cyclicProb$ is undecidable.
\end{theorem}
\begin{proofsketch}
It can be verified that a 2-counter machine $M$ has a periodic configuration iff $\Pi(M)$ contains a loop.
Since the periodic problem of 2-counter machine is undecidable, deciding if an ICG is cyclic is also undecidable.
\end{proofsketch}

\looseness=-1
It can be easily verified that the terminating and cyclic problem for terminating ICGs is decidable since they always terminates and contains no cyclic play.
\begin{theorem} \label{thm:isTerminatingCyclicICG}
	$\termProb$ and $\cyclicProb$ for terminating ICGs is decidable.
\end{theorem}


\section{The Winning Formula Problem}
In this section, we show that the winning formula problem are undecidable.
It becomes decidable for terminating ICGs.

To this end, we first provide another reduction from 2-counter machines to ICGs. 

\begin{definition} \rm \label{def:secondEncoding}
Let $M$ be a 2-counter machine and $L$ the list of labeled instructions of $M$.
An ICG $\Pi'(M) = \tuple{\var, \act, \constraint, \estate}$ is defined as
\begin{itemize}
\item $\statevar: \set{v}$;
\item $\act: \set{a_l \mid 1 \leq l \leq \len[L]  \text{ and } l: \halt \notin L} \cup \set{a_{\len[L] + 1}}$; 

\item $\constraint: v \geq 0 \land \chi_{2, 1, \len[L]}(v) \land \chi_{7, 0, 1}(v)$;

\item $\estate: \left(\bigvee_{l: \halt \in L} \varphi_{2, l}(v) \right) \land \varphi_{7, 1}(v)$.
\end{itemize}	

\noindent
The transition formula $\trans(a_l)$ for action $a_l$ is defined as
\begin{center}
	$\begin{cases}
		\varphi_{2, l}(v) \land \varphi_{7, 1}(v) \land (2^l 7 v' = 2^j 3 v) & \text{if } l: \inc[1, j] \in L; \\
		
		\varphi_{2, l}(v) \land \varphi_{7, 1}(v) \land (2^l 7 v' = 2^j 5 v) & \text{if } l: \inc[2, j] \in L; \\
		
		\begin{aligned}
			&\varphi_{2, l}(v) \land \varphi_{7, 1}(v) \land \\
			&\quad [(\neg \varphi_{3, 0}(v) \land 2^l 3 \cdot 7 v' = 2^j v) \lor \\
			&\quad \;(\varphi_{3, 0}(v) \land 2^l 7 v' = 2^{j'} v)]
		\end{aligned}
		& \text{if } l: \dec[1, j, j'] \in L; \\
		
		\begin{aligned}
			&\varphi_{2, l}(v) \land \varphi_{7, 1}(v) \land \\
			&\quad [(\neg \varphi_{5, 0}(v) \land 2^l 5 \cdot 7 v' = 2^j v) \lor \\
			&\quad\;(\varphi_{5, 0}(v) \land 2^l 7 v' = 2^{j'} v)]\\
		\end{aligned}
		& \text{if } l: \dec[2, j, j'] \in L; \\
		
		\varphi_{2, l}(v) \land \varphi_{7, 0}(v) \land v' = 7 v & \text{if } l = \len[L] + 1.
	\end{cases}$
\end{center}
\end{definition}

Definition \ref{def:secondEncoding} is a slight modification of Definition \ref{def:firstEncoding}.
In Definition \ref{def:secondEncoding}, a state $s$ denotes a configuration $C$, iff $v(s) = 2^{C(r_0)} 3^{C(r_1)} 5^{C(r_2)} 7^i b$, where $0 \leq i \leq 1$ and $b$ is not a multiple of $2$, $3$, $5$ or $7$.
The constraint $\constraint$ requires the value $v(s)$ of every legal state $s$ to be of the form $2^{l} 3^{m} 5^{n} 7^i b$ where $1 \leq l \leq \len[L]$.
The termination condition $\estate$ stipulates that the value $v(s)$ of every terminal state $s$ is $2^{l} 3^{m} 5^{n} 7^1 b$ where $l: \halt \in L$. 
In addition, each action $a_l$ in the ICG $\Pi'(M)$ not only emulates the $l$-th instruction of $L$ but also clear the order of $7$ to $0$, which transits an integer $2^{C(r_0)} 3^{C(r_1)} 5^{C(r_2)} 7^1 b$ to another one $2^{M(C)(r_0)} 3^{M(C)(r_1)} 5^{M(C)(r_2)} 7^0 b$.
The last action $a_{\len[L] + 1}$ only increments the order of $7$ and keep the configuration implicitly expressed in $v(s)$ unchanged, which transits an integer $2^{C(r_0)} 3^{C(r_1)} 5^{C(r_2)} 7^0 b$ to $2^{C(r_0)} 3^{C(r_1)} 5^{C(r_2)} 7^1 b$.

\begin{lemma} \label{lem:isWinForm}
	Let $b$ be a positive integer s.t. $b$ is not a multiple of $2$, $3$, $5$ or $7$.
	
	\begin{enumerate}		
		\item If $s$ is a state such that $v(s) = 2^{l} 3^{m} 5^{n} 7^0 b$ and $l: \halt \notin L$, then there is a unique successor state $s'$ of $s$ s.t. $v(s) = 2^{l} 3^{m} 5^{n} 7^1 b$.
		
		\item If $s$ is a state such that $v(s) = 2^{C(r_0)} 3^{C(r_1)} 5^{C(r_2)} 7^1 b$ and $C(r_0): \halt \notin L$, then there is a unique successor state $s'$ of $s$ s.t. $v(s') = 2^{M(C)(r_0)} 3^{M(C)(r_1)} 5^{M(C)(r_2)} 7^0 b$.
		
		\item For any legal state $s$, there is a unique play from $s$.
		
		\item If $s$ is a state of $\Pi'(M)$ s.t. $v(s) = 2^{l} 3^{m} 5^{n} 7^0 b$ (resp. $2^l 3^{m} 5^{n} 7^{1} b$), then $s$ is a winning (resp. losing) state iff there is a finite play from $s$.
		
		\item If $s$ is a state of $\Pi'(M)$, then $s$ is a draw state iff there is an infinite play from $s$.
	\end{enumerate}
\end{lemma}

\begin{theorem} \label{thm:isWinForm}
$\winFormProb$ is undecidable.
\end{theorem}
\begin{proof}
Let $\phi$ be the formula $v \geq 0 \land \chi_{2, 1, \len[L]}(v) \land \varphi_{7, 0}(v)$. 
The formula $\phi$ exactly characterizes the set of positive integers $2^l 3^m 5^n 7^0 b$ where $1 \leq l \leq \len[L]$ and $b$ is not a multiple of $2$, $3$, $5$ or $7$.
We hereafter prove that $\phi$ is the winning formula for $\Pi'(M)$ iff $M$ is terminating on all configurations.

($\Rightarrow$): Suppose that $\phi$ is the winning formula.
Let $C$ be an configuration of $M$ and $s$ a state s.t. $v(s) = 2^{C(r_0)} 3^{C(r_1)} 5^{C(r_2)} 7^0 b$.
Since $s \models \phi$, $s$ is a winning state.
By Items 3 and 4 of Lemma \ref{lem:isWinForm}, there is a unique finite play $\play$ from $s$ to a terminal state.
By Items 1 and 2 of Lemma \ref{lem:isWinForm} $\play$ alternates between winning states and losing states, and emulates the run of $M$ on $C$.
Hence, $M$ is terminating on all configurations.

($\Leftarrow$): Suppose that $M$ is terminating on all configurations.
Let $s$ be a state s.t. $s \models \constraint \land \phi$.
It follows that $s$ corresponds to a configuration $C$ where $v(s) = 2^{C(r_0)} 3^{C(r_1)} 5^{C(r_2)} 7^0 b$.
By Item 3 of Lemma \ref{lem:isWinForm}, there is a unique play $\play$ from $s$.
By Items 1 and 2 of Lemma \ref{lem:isWinForm}, $\play$ alternates between winning states and losing states, and emulates the run of $M$ on $C$.
Hence, $\play$ eventually ends at a terminal state.
By Item 4 of Lemma \ref{lem:isWinForm}, $s$ is a winning state.
Similarly, if $s$ is a state s.t. $s \models \constraint \land \neg \phi$, then $s$ is a losing state.
Hence, the formula $\phi$ exactly characterizes the set of winning states.

Since the mortality problem of 2-counter machines is undecidable, deciding if an LIA-formula is a winning formula for an LIA-definable ICG is undecidable.
\end{proof}

\looseness=-1
Since no draw state exists in terminating ICG, every legal state has only two categories: winning and losing \cite{Fer2018}.
In addition, the sets of winning states and of losing states are disjoint.
If a formula $\phi$ is the winning formula, then its negation $\neg \phi$ must be the losing formula.
We adopt the three constraints for verifying if a formula $\phi$ is the winning formula proposed in \cite{WuFXLQCR2020}.
Each of the above constraints corresponds to each condition of winning and losing states (cf. Definition \ref{def:wlState}) of terminating ICGs.

\begin{definition}\label{def:consWinForm} \rm 
	Let $\Pi = \tuple{\var, \act, \constraint, \estate}$ be a terminating ICG.
	The constraints for the winning formula $\phi$ of $\Pi$ are
	\begin{enumerate}
		\item $\constraint \land \estate \rightarrow \neg \phi$;
		\item $(\constraint \land \phi) \rightarrow \exists \statevar' \{\trans(\act) \land (\constraint \land \neg \phi)[\statevar/\statevar']\}$;
		\item $(\constraint \land \neg \phi) \rightarrow \forall \statevar' \{\trans(\act) \rightarrow (\constraint \land \phi)[\statevar/\statevar']\}$.
	\end{enumerate}
\end{definition}

\looseness=-1
The following lemma confirms the correctness of the constraints for the winning formula.
\begin{lemma} \label{lem:consWinForm} 
	Let $\Pi = \tuple{\var, \act, \constraint, \estate}$ be a terminating ICG.
	A formula $\phi$ is the winning formula of $\Pi$ iff all of the constraints illustrated in Definition \ref{def:consWinForm} are valid.
\end{lemma}

\begin{theorem} \label{thm:isWinFormTermICG}
	$\winFormProb$ for terminating ICGs is decidable.
\end{theorem}
\begin{proofsketch}
	Lemma \ref{lem:consWinForm} reduces the winning formula problem for a terminating ICG to the validity problem of LIA.
	Thanks to the decidability of the latter problem \cite{Pre1929}, the former decision problem is decidable. 
\end{proofsketch}

\section{The Winning State Problem}
In this section, we show that the winning state problem is undecidable even for terminating IGCs.
We define a class of ICGs, namely finite-depth.
The winning state problem becomes decidable for terminating and finite-depth ICGs.

We hereafter provide the 3rd reduction from 2-counter machines to ICGs.

\begin{definition} \rm \label{def:thirdEncoding}
	Let $M$ be a 2-counter machine and $L$ the list of labeled instructions of $M$.
	An ICG $\Pi''(M) = \tuple{\var, \act, \constraint, \estate}$ is defined as
	\begin{itemize}
		\item $\statevar: \set{v_1, v_2}$;
		
		\item $\act: \set{a_0} \cup \set{a_l, a'_l \mid 1 \leq l \leq \len[L]  \text{ and } l: \halt \notin L}$; 
		
		\item $\constraint: v_1 \geq 0 \land \chi_{2, 1, |L|}(v_1) \land \chi_{7, 0, 2}(v_1) \land v_2 \geq 0$;
		
		\item $\estate\!:\! [(\bigvee_{l: \halt \in L} \varphi_{2, l}(v_1)) \!\land\! \chi_{7, 1, 2}(v_1)] \lor [\varphi_{7, 2}(v_1) \!\land\! v_2 = 0]$.
	\end{itemize}
\noindent
The transition formula $\trans(a_l)$ for action $a_l$ is defined as 
\begin{center}
		$\begin{cases}
				\varphi_{7, 0}(v_1) \land v'_1 = 7 v_1 \land v'_2 \geq 0 & \text{if } l = 0; \\
				
			\begin{aligned}
				&\varphi_{2, l}(v_1) \land \varphi_{7, 2}(v_1) \land v_2 > 0 \land \\
				&\quad (2^l 7 v'_1 = 2^j 3 v_1) \land v'_2 = v_2 - 1
			\end{aligned}
				& \text{if } l: \inc[1, j] \in L; \\
				
			\begin{aligned}
				&\varphi_{2, l}(v_1) \land \varphi_{7, 2}(v_1) \land v_2 > 0 \land \\
				&\quad (2^l 7 v'_1 = 2^j 5 v_1) \land v'_2 = v_2 - 1\\
			\end{aligned}
				& \text{if } l: \inc[2, j] \in L; \\
			
			\begin{aligned}
				&\varphi_{2, l}(v_1) \land \varphi_{7, 2}(v_1) \land v_2 > 0 \land \\
				&\quad [(\neg \varphi_{3, 0}(v_1) \land 2^l 3 \cdot 7 v'_1 = 2^j v_1) \lor \\
			&\quad\; (\varphi_{3, 0}(v_1) \land 2^l 7 v'_1 = 2^{j'} v_1)] \land \\
				&\quad v'_2 = v_2 - 1
			\end{aligned}
				& \text{if } l: \dec[1, j, j'] \in L; \\
				
			\begin{aligned}
				&\varphi_{2, l}(v_1) \land \varphi_{7, 2}(v_1) \land \\
				&\quad [(\neg \varphi_{5, 0}(v_1) \land 2^l 5 \cdot 7 v'_1 = 2^j v_1) \lor {}\\
			&\quad\;(\varphi_{5, 0}(v_1) \land 2^l 7 v'_1 = 2^{j'} v_1)] \land \\
				&\quad v'_2 = v_2 - 1
			\end{aligned}
				&\text{if } l: \dec[2, j, j'] \in L; \\
					
		\end{cases}$
	\end{center}
\end{definition}
\noindent
and $\trans(a'_l)$ for action $a'_l$ is defined as $\varphi_{2, l}(v_1) \land \varphi_{7, 1}(v_1) \land v'_1 = 7 v_1 \land v'_2 = v_2$.	

Definition \ref{def:thirdEncoding} is a slight modification of Definition \ref{def:secondEncoding}.
The ICG $\Pi''(M)$ has two state variables $v_1$ and $v_2$.
The first state variable $v_1$ has the same meaning as $v$ in $\Pi'(M)$ where $v_1(s)$ denotes a configuration $C$, iff $z_1 = 2^{C(r_0)} 3^{C(r_1)} 5^{C(r_2)} 7^i b$ and $z_2 \geq 0$, where $0 \leq i \leq 2$ and $b$ is not a multiple of $2$, $3$, $5$ or $7$.
When $i = 1$ or $2$, the second integer $v_2(s)$ denotes the remaining step of the instruction list $L$ to be executed.
The constraint $\constraint$ requires the value $v_1(s)$ of every legal state $s$ to be of the form $2^{l} 3^{m} 5^{n} 7^i b$ where $1 \leq l \leq \len[L]$ and the value $v_2(s)$ to be non-negative.
The termination condition $\estate$ stipulates that the value $v_1(s)$ of every terminal state $s$ is $2^{l} 3^{m} 5^{n} 7^i b$ where $l: \halt \in L$ and $1 \leq i \leq 2$.
We say a state $s$ is an \textit{initial} state of $\Pi''(M)$, iff $v_1(s) = 2^{l} 3^{m} 5^{n} 7^0 b$ where $1 \leq l \leq \len[L]$; otherwise it is \textit{non-initial}.
The ICG $\Pi''(M)$ satisfies the following properties:

\begin{lemma} \label{lem:isWinState}
	\begin{enumerate}		
		\item If $s$ is a state such that $v_1(s) = 2^{l} 3^{m} 5^{n} 7^1 b$ and $l: \halt \notin L$, then there is a unique successor state $s'$ of $s$ s.t. $v_1(s) = 2^{l} 3^{m} 5^{n} 7^2 b$ and $v_2(s') = v_2(s)$.
		
		\item If $s$ is a state such that $v_1(s) = 2^{C(r_0)} 3^{C(r_1)} 5^{C(r_2)} 7^2 b$, $v_2(s) > 0$ and $C(r_0): \halt \notin L$, then there is a unique successor state $s'$ of $s$ s.t. $v_1(s') = 2^{M(C)(r_0)} 3^{M(C)(r_1)} 5^{M(C)(r_2)} 7^1 b$ and $v_2(s') = v_2(s) \!-\! 1$.
			
		\item Every $\play$ of $\Pi''(M)$ is finite.
		
		\item If $s$ is a state on a play $\play$ from an initial state s.t. $v_1(s) = 2^{l} 3^{m} 5^{n} 7^1 b$, then $s$ is losing (resp. winning) iff $\play$ ends at a state $s^*$ s.t. $v_1(s^*) = 2^{l^*} 3^{m^*} 5^{n^*} 7^1 b$ (resp. $v_1(s^*) = 2^{l^*} 3^{m^*} 5^{n^*} 7^2 b$).
		
		\item If $s$ is a state on a play $\play$ from an initial state s.t. $v_1(s) = 2^{l} 3^{m} 5^{n} 7^2 b$, then $s$ is winning (resp. losing) iff $\play$ ends at a state $s^*$ s.t. $v_1(s^*) = 2^{l^*} 3^{m^*} 5^{n^*} 7^1 b$ (resp. $v_1(s^*) = 2^{l^*} 3^{m^*} 5^{n^*} 7^2 b$).
	\end{enumerate}
\end{lemma}

\begin{theorem} \label{thm:isWinState}
	$\winStateProb$ (even for terminating ICGs) is undecidable.
\end{theorem}
\begin{proof}
	By Item 3 of Lemma \ref{lem:isWinState}, every play of $\Pi''(M)$ is finite and hence $\Pi''(M)$ is terminating.

	Let $C$ be a configuration of a 2-counter machine $M$ and $s$ a state of $\Pi''(M)$ s.t. $v_1(s) = 2^{C(r_0)} 3^{C(r_1)} 5^{C(r_2)} 7^0 b$.	
	We hereafter prove that the state $s$ is winning iff $M$ is terminating in $C$.
	
	($\Rightarrow$):
	Suppose that $s$ is winning.
	By the definition of winning and losing states (cf. Definition \ref{def:wlState}), there is a losing successor state $s'$ of $s$.
	According to the transition formula $\trans(a_l)$, $v_1(s') = 2^{l} 3^{m} 5^{n} 7^1 b$.
	$s'$ is on a play $\play$ from $s$.	
	By Items 4 and 5 of Lemma \ref{lem:isWinState}, $\play$ ends at a terminal state $s^*$ s.t. $v_1(s^*) = 2^{l^*} 3^{m^*} 5^{n^*} 7^1 b$ where $l^*: \halt \in L$.
	By Items 1 and 2 of Lemma \ref{lem:isWinState}, $\Pi''(M)$ emulates the behavior of $M$.
	Hence, $M$ is terminating in $C$.
	
	($\Leftarrow$): We prove by contrapositive.	
	Suppose that $s$ is not winning.
	Since $\Pi''(M)$ is terminating, it has no draw states, and hence $s$ is losing.
	By the definition of winning and losing states, every successor state $s'$ of $s$ is winning.	
	According to the transition formula $\trans(a_l)$, $v_1(s') = 2^{l} 3^{m} 5^{n} 7^1 b$.
	$s'$ is on a play $\play$ from $s$.	
	By Items 4 and 5 of Lemma \ref{lem:isWinState}, $\play$ ends at a terminal state $s^*$ s.t. $v_1(s^*) = 2^{l^*} 3^{m^*} 5^{n^*} 7^2 b$ where $l^*: \halt \notin L$ and $v_2(s^*) = 0$.
	By Items 1 and 2 of Lemma \ref{lem:isWinState}, $\Pi''(M)$ emulates the behavior of $M$.
	Hence, $M$ is not terminating in $C$.
	
	Since the halting problem of 2-counter machines is undecidable, deciding if a state is winning for an LIA-definable ICG is undecidable.	
\end{proof}

We define a class of terminating ICGs s.t. its winning state problem becomes decidable.
\begin{definition} \rm \label{def:depth}
	Let $s$ be a winning or losing state of an ICG.
	The depth $d(s)$ of $s$ is defined as: \\ 
		$\begin{cases}
			0 & \text{if } s \text{ is a terminal state}; \\
			
			\min \set{d(s') \mid s' \in \loseSet \cap \suc(s)}  + 1 & \text{if } s \text{ is a winning state}; \\
			
			\max \set{d(s') \mid s' \in \suc(s)} + 1 & \text{otherwise}. 
		\end{cases}$
	where $\min(S)$ is the least element of $S$ and $\max(S)$ is the greatest element of $S$.
\end{definition}

We say a state $s$ is of depth $i$, iff $\dep[s] = i$.
The depth $\dep[s]$ is undefined for some state.
We illustrate this in the following example.

\begin{example} \label{exm:finDep}
	An ICG $\Pi = \tuple{\var, \act, \constraint, \estate}$ is defined as
	\begin{itemize}
		\item $\statevar: \set{v_1, v_2}$ and $\act: \set{a}$;
		
		\item $\constraint: (v_1 = 0 \land v_2 \geq 0) \lor (v_1 = 1 \land v_2 = 0) $;
		\item $\estate: v_1 = 0 \land v_2 = 0$;
		
		\item $\trans(a): (v_1 = 0 \land v_2 > 0 \land v'_1 = v_1 \land v'_2 = v_2 - 1) \lor \\$ $\hspace*{10.5mm} (v_1 = 1 \land v_2 = 0 \land v_1' = 0 \land v'_2 \modulo_2 1)$.
	\end{itemize}
	
	We use a pair $(v_1(s), v_2(s))$ for a state $s$.
	It is easily verified that every state $(0, i)$ where $i$ is even (resp. odd) is losing (resp. winning).
	Since state $(1, 0)$ has a unique successor state $(0, i)$ where $i$ is odd, $(1, 0)$ is losing.	
	By the definition of depth, we obtain that $\dep[(0, i)] = i$.
	However, the set $S$ of the depth of every states s.t. $i \in \set{1, 3, \cdots}$ and has no greatest element.
	Hence, $\dep[(1, 0)]$ is undefined.
	\qed
\end{example}

We say an ICG is \textit{finite-depth}, iff the depth of every winning or losing state $s$ is defined.

\begin{algorithm}[!t]
	\small
	\caption{${\tt IsWinningState}(\Pi, s)$}
	\label{alg:isWinState}
	\KwIn{$\Pi = \tuple{\var, \act, \constraint, \estate}$: an LIA-definable ICG; \\
		\hspace*{9.5mm}$s$: a legal state.}
	\KwOut{$\true$, if $s$ is a winning state; $\false$, otherwise.}
	
	$\varphi_{0} \assign \false$ and $\psi_{0} \assign (\constraint \land \estate)$ 
	
	\lIf{$s \models \psi_{0}$}
	{
		\Return{$\false$}
	}
	\Else{
		$i \assign 1$
		
		\While{{\bf true}}
		{
			$\varphi_{i} \assign$ a quantifier-free formula equivalent to $\hspace*{6.5mm} \exists \statevar' \{\trans(\act) \land (\constraint \land \psi_{i - 1})[\statevar/\statevar']\}$

			$\psi_{i} \assign$ a quantifier-free formula equivalent to $\hspace*{6.5mm} \forall \statevar' \{\trans(\act) \rightarrow (\constraint \land \varphi_{i - 1})[\statevar/\statevar']\}$
			
			\lIf{$s \models \constraint \land \varphi_{i}$}
			{
				\Return{$\true$}
			}
			
			\lIf{$s \models \constraint \land \psi_{i}$}
			{
				\Return{$\false$}
			}
			
			$i \assign i + 1$
		}
	}
\end{algorithm}

\looseness=-1
Algorithm \ref{alg:isWinState} iteratively constructs two class formulas: $\psi_i$'s and $\varphi_i$'s, representing the two sets of winning states with depth at most $i$ and of losing states with depth at most $i$, respectively.
Initially, the formula $\varphi_0$ is $\false$, meaning that no terminal state is winning, and $\psi_0$ is $\constraint \land \estate$, corresponding to Condition 1 of Definition \ref{def:wlState} (Line 1).
Then, Algorithm \ref{alg:isWinState} enters a loop (Lines 5 - 10).
According to Condition 2 of Definition \ref{def:wlState}, we construct the formula $\varphi_i$ from the formula $\psi_{i - 1}$, representing the set of winning states that have at least one successor losing state with depth at most $i - 1$ (Line 6).
The notation $(\constraint \land \psi_{i - 1})[\statevar/\statevar'])$ denotes the formula obtained by replacing every occurrence of $v \in \statevar$ in $\constraint \land \psi$ with $v'$.
Similarly, $\psi_i$ denotes the subset of losing states all of whose successor states satisfying $\varphi_{i - 1}$  (Line 7).
If $s$ satisfies $\psi_i$, then $s$ is a losing state (Lines 2 \& 9).
Similarly, $s$ is a winning state when $s \models \varphi_i$ (Line 8).


Algorithm \ref{alg:isWinState} enjoys the termination, soundness and completeness properties.

\begin{lemma} \label{lem:isWinStateAlg}
	Let $\Pi$ be a terminating finite-depth ICG and $s$ a legal state.
	Then, 
	\begin{itemize}		
		\item Termination: Algorithm ${\tt IsWinningState}(\Pi, s)$ always terminates.
		
		\item Soundness: If ${\tt IsWinningState}(\Pi, s)$ returns $\true$, then $s$ is a winning state.
		
		\item Completeness: If $s$ is a winning state, then ${\tt IsWinningState}(\Pi, s)$ returns $\true$.		
	\end{itemize}	
\end{lemma}

As a corollary of Lemma \ref{lem:isWinStateAlg}, we obtain the following decidability result.
\begin{theorem} \label{thm:isWinStateTermFD}
	$\winStateProb$ for terminating finite-depth ICGs is decidable.
\end{theorem}

\begin{algorithm}[!t]
	\small
	\caption{${\tt IsWinningState'}(\Pi, s)$}
	\label{alg:isWinStateMiniMax}
	\KwIn{$\Pi = \tuple{\var, \act, \constraint, \estate}$: an LIA-definable ICG; \\
		\hspace*{9.5mm}$s$: a legal state.}
	\KwOut{$\true$, if $s$ is a winning state; $\false$, otherwise.}
	
	\lIf{$s \models \constraint \land \estate$}
	{		
		\Return{$\false$}
	}
	
	\ForEach{legal successor state $s'$ of $s$}
	{
		\lIf{${\tt IsWinningState'}(\Pi, s') = \false$}
		{
			\Return{$\true$}
		}
	}
	
	\Return{$\false$}
\end{algorithm}

Finally, we compare Algorithm \ref{alg:isWinState} to the minimax-style algorithm, depicted in Algorithm \ref{alg:isWinStateMiniMax} and illustrated in \cite{BelM2020}. 
Algorithm \ref{alg:isWinStateMiniMax} is a forward method that starts from a state $s$ and then obtain all of descendant states until it reaches any terminal state.
It decides if $s$ is winning based on its successor states according to Definition \ref{def:wlState}.
If $s$ is a terminal state, then it is losing (Line 1).
If at least one losing successor state of $s$ exits, then $s$ is winning (Lines 2 \& 3).
Otherwise, it is losing (Line 4).
Compared to Algorithm \ref{alg:isWinState}, the minimax-style algorithm does not terminates when a losing state has infinitely many successor state.
We illustrate this in the following example.


\begin{example}
	\looseness=-1
	We consider the ICG that is the same as Example \ref{exm:finDep} except the transition formula:
	\begin{itemize}
		\item $\trans(a): (v_1 = 0 \land v_2 > 0 \land v'_1 = 0 \land v'_2=0) \lor \\$ $\hspace*{10.5mm} (v_1 = 1 \land v_2 = 0 \land v'_1 = 0 \land v'_2 > 0)$.
	\end{itemize}
	
%
		Then state $(1, 0)$ has infinitely many successor states $(0, i)$ for $i > 0$.
		Such successor state $(0, i)$ has a unique successor $(0, 0)$, which is a terminal state, and hence is a winning state.
		It follows that $(1, 0)$ is losing.
		The loop (Lines 2 - 3) in Algorithm \ref{alg:isWinStateMiniMax} does not terminate for state $(1, 0)$.
		
		It first constructs the formula $\psi_0: v_1 = 0 \land v_2 = 0$, and then the formula $\varphi_1: v_1 = 0 \land v_2 > 0$ that contains all states $(0, i)$ with $i > 0$, and finally the formula $\psi_2: (v_1 = 0 \lor v_1 = 1) \land v_2 = 0$ that contains the state $(1, 0)$.
		Thus, Algorithm \ref{alg:isWinState} terminates with the label $\false$, meaning $(1, 0)$ is losing.
		\qed
\end{example}

\section{The Losing Formula and State Problems} \label{sec:loseStateFormProb}
In this section, we extend Theorems \ref{thm:isWinForm}, \ref{thm:isWinFormTermICG}, \ref{thm:isWinState} and \ref{thm:isWinStateTermFD} to losing states and formulas.
%

\begin{theorem} \label{thm:isLoseForm}
	$\loseFormProb$ is undecidable.
\end{theorem}
\begin{proofsketch}
	Let $\phi$ be the formula $v \geq 0 \land \chi_{2, 1, \len[L]}(v) \land \varphi_{7, 1}(v)$. 
	The formula $\phi$ exactly characterizes the set of positive integers $2^l 3^m 5^n 7^1 b$ where $1 \leq l \leq \len[L]$ and $b$ is not a multiple of $2$, $3$, $5$ or $7$.
	Similarly to the proof of Theorem \ref{thm:isWinForm}, $\phi$ is the losing formula for $\Pi'(M)$ iff $M$ is terminating on all configurations.
	Since the mortality problem of 2-counter machines is undecidable, deciding if an LIA-formula is the losing formula for a ICG is undecidable.
\end{proofsketch}

\begin{theorem} \label{thm:isLoseFormTermICG}
	$\loseFormProb$ for terminating ICGs is decidable.
\end{theorem}
\begin{proof}
	In terminating ICGs, the negation of the winning formula is the losing formula, and vice versa.
	We reduce the losing formula problem to the winning formula problem for terminating ICGs.
	By Theorem \ref{thm:isWinFormTermICG}, the latter decision problem is decidable.
	So is the former decision problem.
\end{proof}

\begin{theorem} \label{thm:isLoseState}	
	$\loseStateProb$ (even for terminating ICGs) is undecidable.
\end{theorem}
\begin{proofsketch}
	For terminating ICGs, $\winSet$ and $\loseSet$ are disjoint.
	It follows that $\loseStateProb$ is the complement of $\winStateProb$ and hence $\loseStateProb$ is also undecidable.
%
%
%
\end{proofsketch}

\begin{theorem} \label{thm:isLoseStateTermICG}
	$\loseStateProb$ for terminating finite-depth ICGs is decidable.
\end{theorem}
\begin{proofsketch}
	It can be verified that Algorithm ${\tt IsWinningState}(\Pi, s)$ returns $\false$ iff $s$ is a losing state.
	Hence, Algorithm \ref{alg:isWinState} is also complete and sound for losing states.
	It follows that deciding a state is losing for a terminating finite-depth ICG is decidable.
\end{proofsketch}

\section{The Draw Formula and State Problems} \label{sec:extensionToLoseDraw}
In this section, we analyze the decidability of the draw state and formula problems. 
%

\looseness=-1
Just like the winning and losing formula problems, the draw formula problem for general ICGs is undecidable, but becomes decidable for terminating ICGs.
\begin{theorem} \label{thm:isDrawForm}
	$\drawFormProb$ is undecidable.
\end{theorem}
\begin{proof}
	Let $M$ be a 2-counter machine and $\Pi'(M)$ the ICG defined in Definition \ref{def:secondEncoding}.	
	Let $\phi_d$ be the formula $\false$.
	The formula $\phi_d$ denotes the empty set of integers.
	We hereafter prove that $\phi_d$ is the draw formula of $\Pi'(M)$ iff $M$ is terminating on all configurations.
	
	($\Rightarrow$): Suppose that $\phi_d$ is the draw formula.
	It follows that no draw state exists.	
	We prove $M$ is terminating on all configurations by contradiction.
	Assume that $M$ is not terminating on a configuration $C$.
	Let $s$ be a state s.t. $v(s) = 2^{C(r_0)} 3^{C(r_1)} 5^{C(r_2)} 7^0 b$.
	By Item 3 of Lemma \ref{lem:isWinForm}, there is a unique play $\play$ from $s$.
	By Items 1 and 2 of Lemma \ref{lem:isWinForm}, $\play$ alternates between winning states and losing states, and emulates the run of $M$ on $C$.
	Hence, $\play$ is infinite.
	By Items 5 of Lemma \ref{lem:isWinForm}, $s$ is a draw state.
	This contradicts the assumption that no draw state exists.		
	Hence, $M$ is terminating on all configurations.	
	
	($\Leftarrow$): Suppose that $M$ is terminating on all configurations.
	We remind that the formula $\constraint$ for legal states is $v \geq 0 \land \chi_{2, 1, \len[L]}(v) \land \chi_{7, 0, 1}(v)$.
	Let $\phi_w$ be the formula $v \geq 0 \land \chi_{2, 1, \len[L]}(v) \land \varphi_{7, 0}(v)$ and $\phi_l$ be $v \geq 0 \land \chi_{2, 1, \len[L]}(v) \land \varphi_{7, 1}(v)$.
	By the proof of Theorems \ref{thm:isWinForm} and \ref{thm:isLoseForm}, we get that $\phi_w$ and $\phi_l$ is the winning and losing formula for $\Pi'(M)$, respectively.
	Clearly, $\phi_w \lor \phi_l \equiv \constraint$.
	In addition, the three sets $\winSet$, $\loseSet$ and $\drawSet$ are mutually disjoint.
	It follows that the set of draw states is empty and hence $\phi_d$ is the draw formula.
	
	Since the mortality problem of 2-counter machines is undecidable, deciding if an LIA-formula is the draw formula for a ICG is undecidable.
\end{proof}

\begin{theorem} \label{thm:isDrawFormTermICG}
	$\drawFormProb$ for terminating ICGs is decidable.
\end{theorem}
\begin{proof}
	Since terminating ICGs have no draw state, the draw formula is $\false$.
	It follows that an LIA-formula $\phi$ is the draw formula iff $\neg \phi$ is valid.
	Hence, deciding if an LIA-formula is the draw formula for a terminating ICG is decidable.
\end{proof}

Unlike the winning and losing state problems, the draw state problem for terminating ICGs is decidable, but is still undecidable for general ICGs.
We obtain the undecidability result via the 2nd reduction $\Pi'(M)$.
\begin{theorem} \label{thm:isDrawState}
	$\drawStateProb$ is undecidable.
\end{theorem}
\begin{proof}
	Let $C$ be a configuration of a 2-counter machine $M$ and $s$ a state of $\Pi'(M)$ s.t. $v(s) = 2^{C(r_0)} 3^{C(r_1)} 5^{C(r_2)} 7^0 b$.
	We hereafter prove that the state $s$ is draw iff $M$ is not terminating in $C$.
	We only verify the if direction.
	The only-if direction can be similarly proved.
	
	Suppose that $s$ is draw.
	By Item 3 of Lemma \ref{lem:isWinForm}, there is a unique play $\play$ from $s$.
	By Item 5 of Lemma \ref{lem:isWinForm}, $\play$ is infinite.
	By Items 1 and 2 of Lemma \ref{lem:isWinForm}, $\play$ emulates the behavior of $M$.
	Hence, $M$ is not terminating in $C$.
	
	
	The halting problem of 2-counter machines is undecidable.
	So is its complement.
	Hence, deciding if a state is draw for an LIA-definable ICG is undecidable.
\end{proof}

\begin{theorem} \label{thm:isDrawFormulaStateTermICG}
	$\drawStateProb$ for terminating ICGs is decidable.
\end{theorem}
\begin{proofsketch}
	It follows directly from that any terminating ICG has no draw state. 
\end{proofsketch}

\section{Conclusions}

\looseness=-1
In this paper, we have studied several decision problems for LIA-definable ICGs, including $\termProb$, $\cyclicProb$, $\winFormProb$, $\loseFormProb$, $\drawFormProb$, $\winStateProb$, $\loseStateProb$ and $\drawStateProb$.
Table \ref{tab:summary-results} summarizes the decidability and undecidability results of the above problems for arbitrary ICGs, terminating ICGs, and terminating finite-depth ICGs.

%
%

\looseness=-1
To prove the undecidability result, we design three reductions $\Pi(M)$, $\Pi'(M)$ and $\Pi''(M)$ from 2-counter machines to LIA-definable ICGs.
Firstly, we utilize the prime encoding to represent each configuration $C$ of the 2-counter machine $M$ by a value of a state variable in each reduction.
Secondly, we design the transition formula of each action to emulate the behavior of $M$.
Finally, we reduce the halting, mortality or periodic problems of 2-counter machines to the above decision problems of ICGs.
The halting, mortality and periodic problems of 2-counter machines are all undecidable.
So are the above decision problems for arbitrary ICGs, which is shown in the first column of Table \ref{tab:summary-results}.

\looseness=-1
The terminating condition removes infinite plays together with draw states from ICGs.
This directly yields the decidability of $\termProb$, $\cyclicProb$, $\drawStateProb$ and $\drawFormProb$ for terminating ICGs.
It also makes the winning and losing formula problems decidable.
The absence of draw states, the correctness of a candidate winning (or losing) formula can be checked via the validity of several constraints that are formalized in LIA and that are based on the recursive definition of winning and losing states.
However, the terminating condition is not sufficient for the decidability of $\winStateProb$ and $\loseStateProb$, as demonstrated via the 3rd reduction that still emulates the behavior of a given 2-counter machine.

\looseness=-1
The finite-depth condition ensures that every winning or losing state has a well-defined integer-valued depth.
Under the finite-depth condition as well as the terminating condition, Algorithm \ref{alg:isWinState} constructs the two class of formulas that exactly captures the two sets of winning states and of losing states with monotonically increasing depth, respectively.
These two class of formulas eventually covers every winning or losing state for terminating finite-depth ICGs and therefore establishing the decidability of $\winStateProb$ and $\loseStateProb$. 

\looseness=-1
We have extended the above results to the mis\`{e}re play rule via some minor adjustments of their proof. 
Appendix \ref{sec:extensionToMisere} presents these adjustments.


\looseness=-1

\begin{table}
	\centering
	\small
	\caption{Summary of undecidability and decidability results for LIA-definable ICGs where \textbf{U} and \textbf{D} denote undecidable and decidable, respectively, and T$n$ denotes Theorem $n$.}
	\label{tab:summary-results}
	\begin{tabular}{|c|c|c|c|}
		\hline
		\begin{tabular}[c]{@{}c@{}}Decision \\ Problems\end{tabular}
		& \begin{tabular}[c]{@{}c@{}}Arbitrary\end{tabular}
		& \begin{tabular}[c]{@{}c@{}}Terminating\end{tabular}
		& \begin{tabular}[c]{@{}c@{}}Terminating \& \\ Finite-depth\end{tabular} \\
		\hline
		\hline
		\termProb
		& \textbf{U} [T\ref{thm:isTerminatingICG}]
		& \textbf{D} [T\ref{thm:isTerminatingCyclicICG}]
		& \textbf{D} [T\ref{thm:isTerminatingCyclicICG}] \\
		\hline
		\cyclicProb
		& \textbf{U} [T\ref{thm:isCyclicICG}]
		& \textbf{D} [T\ref{thm:isTerminatingCyclicICG}]
		& \textbf{D} [T\ref{thm:isTerminatingCyclicICG}]\\
		\hline
		\winFormProb
		& \textbf{U} [T\ref{thm:isWinForm}]
		& \textbf{D} [T\ref{thm:isWinFormTermICG}]
		& \textbf{D} [T\ref{thm:isWinFormTermICG}] \\
		\hline
		\loseFormProb
		& \textbf{U} [T\ref{thm:isLoseForm}]
		& \textbf{D} [T\ref{thm:isLoseFormTermICG}]
		& \textbf{D} [T\ref{thm:isLoseFormTermICG}]\\
		\hline
		\drawFormProb
		& \textbf{U} [T\ref{thm:isDrawForm}]
		& \textbf{D} [T\ref{thm:isDrawFormTermICG}]
		& \textbf{D} [T\ref{thm:isDrawFormTermICG}] \\
		\hline
		\winStateProb
		& \textbf{U} [T\ref{thm:isWinState}]
		& \textbf{U} [T\ref{thm:isWinState}]
		& \textbf{D} [T\ref{thm:isWinStateTermFD}] \\
		\hline
		\loseStateProb
		& \textbf{U} [T\ref{thm:isLoseState}]
		& \textbf{U} [T\ref{thm:isLoseState}]
		& \textbf{D} [T\ref{thm:isLoseStateTermICG}] \\
		\hline
		\drawStateProb
		& \textbf{U} [T\ref{thm:isDrawState}]
		& \textbf{D} [T\ref{thm:isDrawFormulaStateTermICG}]
		& \textbf{D} [T\ref{thm:isDrawFormulaStateTermICG}]
		\\
		\hline
	\end{tabular}
\end{table}

\bibliography{references}

\appendix

\section{Proofs of Lemmas}
\begin{lemma}
	$(\winSet \cup \loseSet, \order^+)$ is a well-founded set.
\end{lemma}
\begin{proof}
	By the recursive definition of winning and losing states (cf. Definition \ref{def:wlState}), every play on $\winSet \cup \loseSet$ is finite.
	By Proposition 5.3 in \cite{Levy1979}, $\order^+$ is a well-founded relation on $\winSet \cup \loseSet$.
%
%
%
\end{proof}

\noindent
\textbf{Proof of Lemma \ref{lem:isWinForm}}
\begin{proof}
	\begin{enumerate}
		\item According to the transition formula $\trans(a_l)$, there is a unique action $a_{\len[L] + 1}$ applicable in $s$.
		Hence, $s$ has a unique successor state $s'$ s.t. $v(s') = 2^{l} 3^{m} 5^{n} 7^1 b$.
		
		\item The proof is similar to Item 1 by using action $a_l$ where $1 \leq l \leq \len[L]$ and $l: \halt \notin L$.
		
		\item It is directly from Items 1 and 2.

		\item 
		We prove by induction on $(\winSet \cup \loseSet, \order^+)$.
		
		\textbf{Base case} ($s$ is a minimal element of $\order^+$):
		It follows that $s$ is a legal terminal state.
		In this case, both the two facts that $s$ is a losing state and that the play from $s$ is finite hold.
		
		\textbf{Inductive step} ($s$ is a non-minimal element of $\order^+$):
		We first verify the case where $v(s) = 2^{l} 3^{m} 5^{n} 7^0 b$.
		The other case can be similarly proved.	
		
		($\Rightarrow$):
		Suppose that $s$ is a winning state.
		According to the transition formula $\trans(a_l)$, there is a unique successor state of $s'$ s.t. $v(s') = 2^{l} 3^{m} 5^{n} 7^1 b$.	
		By the definition of winning and losing states (cf. Definition \ref{def:wlState}), $s'$ is a losing state.
		By the inductive hypothesis, there is a finite play from $s'$.
		So is $s$.
		
		($\Leftarrow$): 
		Suppose that there is a finite play from $s$.
		According to the transition formula $\trans(a_l)$, there is a unique successor state of $s'$ s.t. $v(s') = 2^{l} 3^{m} 5^{n} 7^1 b$.	
		It follows that there is a finite play from $s'$.
		By the inductive hypothesis, $s'$ is a winning state.
		Hence, $s$ is a losing state.
		
		\item ($\Rightarrow$):
		Suppose that $s$ is a draw state.
		By Item 3, there is a unique play $\play$ from $s$.
		By Item 4, $\play$ is infinite.
		
		($\Leftarrow$):
		Suppose that there is an infinite play $\play$ from $s$.
		By Item 3, $\play$ is the unique play from $s$.
		By Item 4, $s$ is a draw state.		
	\end{enumerate}	
\end{proof}

\noindent
\textbf{Proof of Lemma \ref{lem:consWinForm}}
\begin{proof}
	($\Rightarrow$): Assume that $\phi$ is the winning formula.
	It follows that every state $s$ satisfying $\constraint \land \phi$ is winning and every state $s$ satisfying $\constraint \land \neg \phi$ is losing.	
	We prove the three constraints illustrated in Definition \ref{def:consWinForm} are valid.
	\begin{enumerate}
		\item $\constraint \land \estate \rightarrow \neg \phi$: 
		Let $s$ be a state satisfying $\constraint \land \estate$.
		It follows that $s$ is a legal terminal state.
		By the first item of the definition of winning and losing states (cf. Definition \ref{def:wlState}), $s$ is a losing state.
		Hence, $s \models \neg \phi$ and the first constraint is valid.
		
		\item $(\constraint \land \phi) \rightarrow \exists \statevar' [\trans(\act) \land (\constraint \land \neg \phi)[\statevar/\statevar']]$: 
		Let $s$ be a state satisfying $\constraint \land \phi$.
		It follows that $s$ is a legal winning state.
		By the second item of the definition of winning and losing states, there is at least one applicable action $a$ in $s$ s.t. $\tau(s, a)$ is a losing state .
		Hence, $s \models \exists \statevar' [\trans(\act) \land (\constraint \land \neg \phi)[\statevar/\statevar']]$ and the second constraint is valid.
		
		\item $(\constraint \land \neg \phi) \rightarrow \forall \statevar' [\trans(\act) \rightarrow (\constraint \land \phi)[\statevar/\statevar']]$:
		The proof of the validity of the third constraint is the same as that of the second constraint by using the third item of the definition of winning and losing states.
	\end{enumerate}
	
	($\Leftarrow$): Assume that $\phi$ is a formula s.t. all of the constraints illustrated in Definition \ref{def:consWinForm} is valid.
	We prove that every state $s$ satisfying $\constraint \land \phi$ is winning and every state $s$ satisfying $\constraint \land \neg \phi$ is losing by induction on $(\winSet \cup \loseSet, \order^+)$.
	
	\textbf{Base case} ($s$ is a minimal element of $\order^+$):
	It follows that $s$ is a legal terminal state and $s \models \constraint \land \estate$.	
	By the first constraint $\constraint \land \estate \rightarrow \neg \phi$, we get that $s \models \neg \phi$.
	By the definition of winning formula (cf. Definition \ref{def:winForm}), $s$ is a losing state.
	
	\textbf{Inductive step} ($s$ is a non-minimal element of $\order$):
	Suppose that $s \models \phi$.
	By the second constraint $(\constraint \land \phi) \rightarrow \exists \statevar' [\trans(\act) \land (\constraint \land \neg \phi)[\statevar/\statevar']]$, there is a successor state $\tau(s, a)$ of $s$ s.t. $\tau(s, a) \models \constraint \land \neg \phi$.
	By the inductive assumption, $\tau(s, a)$ is a losing state.
	By the second item of the definition of winning and losing states, $s$ is a winning state.
	Similarly, we can prove that $s$ is a losing state when $s \models \neg \phi$ by using the third constraint $(\constraint \land \neg \phi) \rightarrow \forall \statevar' [\trans(\act) \rightarrow (\constraint \land \phi)[\statevar/\statevar']]$ and the third item of the definition of winning and losing states.
\end{proof}

\noindent
\textbf{Proof of Lemma \ref{lem:isWinState}}
\begin{proof}
	\begin{enumerate}
		
		\item According to the two transition formulas $\trans(a_l)$ and $\trans(a'_l)$, there is a unique action $a'_l$ applicable in $s$.
		Hence, $s$ has a unique successor state $s'$ s.t. $v_1(s') = 2^{l} 3^{m} 5^{n} 7^2 b$ and $v_2(s') = v_2(s)$.
		
		\item The proof is similar to Item 1 by using action $a_l$ where $1 \leq l \leq \len[L]$ and $l: \halt \notin L$.
		
		\item Let $s$ be a legal state and $\play$ a play from $s$.
		We prove by the following cases:
		
		\begin{enumerate}
			\item $v_1(s) = 2^{l} 3^{m} 5^{n} 7^i b$ where $i \in \set{1, 2}$ and $l: \halt \in L$: By the termination condition $\estate$, no applicable action in $s$ exists.
			Hence, the play from $s$ contains only one terminal state and hence is finite.
			
			\item $v_1(s) = 2^{l} 3^{m} 5^{n} 7^2 b$ and $v_2(s) = 0$:
			The proof is similar to Case a.
			
			\item $v_1(s) = 2^{l} 3^{m} 5^{n} 7^0 b$: 
			Clearly, it is an initial state of $\Pi'(M)$.
			According to the two transition formulas $\trans(a_l)$ and $\trans(a'_l)$ there is a unique action $a_{0}$ applicable in $s$.			
			$s$ has a successor state $s'$ where $v_1(s') = 2^{l} 3^{m} 5^{n} 7^1 b$ and $v_2(s')$ is a natural number in $\play$.
			Let $C$ be the configuration s.t. $C(r_0) = l$, $C(r_1) = m$ and $C(r_2) = n$.
			Suppose that $C$ is not halting and $v_2(s') > 0$.
			By Items 1 and 2, we get the descendant state $t'$ of $s'$ s.t. $v_1(s') = 2^{M(C(r_0))} 3^{M(C(r_1))} 5^{M(C(r_2))} 7^1 b$ and $v_2(t') = v_2(s') - 1$ via performing two actions $a_l$ and $a_{\len[L] + 1}$.
			In the play, the value of $v_2$ decrements and the configuration represented by $v_1$ takes a single-step transition after two actions.			
			If every descendant state $s''$ of $s$ on $\play$ satisfies that $v_1(s'') = 2^{l''} 3^{m''} 5^{n''} 7^i b$ where $l'': \halt \notin L$, then it eventually ends at a terminal state defined as Case b.
			If some descendant state $s''$ of $s$ on $\play$ satisfies that $v_1(s'') = 2^{l''} 3^{m''} 5^{n''} 7^1 b$ where $l'': \halt \in L$, then it eventually ends at a terminal state defined as Case a.
			
			\item $v_1(s) = 2^{l} 3^{m} 5^{n} 7^i b$ and $v_2(s) > 0$ where $i \in \set{1, 2}$ and $l: \halt \notin L$: 
			The play $\play$ from $s$ is a sub-play of $\play'$ from an initial state.
			By Case c, $\play'$ satisfies the condition illustrated in Item 2.
			So does $\play$.
			
			\item $v_1(s) = 2^{l} 3^{m} 5^{n} 7^1 b$ and $v_2(s) = 0$ where $l: \halt \notin L$: 
			The proof is similar to Case d.
		\end{enumerate}
		
		\item We prove Items 4 and 5 together.		
		We here only prove that if $s$ is a state on $\play$ s.t. $v_1(s) = 2^{l} 3^{m} 5^{n} 7^1 b$, then $s$ is losing iff $\play$ ends at a state $s^*$ s.t. $v_1(s^*) = 2^{l^*} 3^{m^*} 5^{n^*} 7^1 b$.
		The other cases can be similarly proved.
		Let $S$ be the set of states on $\play$.
		Since $S$ is a subset of $\winSet \cup \loseSet$, $(S, \order^+)$ is well-founded set.
		We prove by induction $(S, \order^+)$.
		
		\textbf{Base case} ($s'$ is a minimal element of $\order^+$):
		It follows that $s'$ is a legal terminal state.
		By the definition of winning and losing states (cf. Definition \ref{def:wlState}), $s$ is a losing state. 
		
		\textbf{Inductive step} ($s'$ is a non-minimal element of $\order^+$):
		It follows that $l: \halt \notin L$.
		
		($\Rightarrow$):
		Suppose that $s$ is a losing state.		
		By Item 1, there is a unique successor state $s'$ of $s$ s.t. $v_1(s') = 2^{l} 3^{m} 5^{n} 7^2 b$ and $v_2(s') = v_2(s)$.
		By the definition of winning and losing states, $s'$ is a winning state.
		By the inductive hypothesis, the play $\play'$ from $s'$ ends at a state $s^*$ s.t. $v_1(s^*) = 2^{l^*} 3^{m^*} 5^{n^*} 7^1 b$.
		So is $\play$.
		
		($\Leftarrow$): 
		Suppose that $\play$ ends at a state $s^*$ s.t. $v_1(s^*) = 2^{l^*} 3^{m^*} 5^{n^*} 7^1 b$.	
		By Item 1, there is a unique successor state $s'$ of $s$ s.t. $v_1(s') = 2^{l} 3^{m} 5^{n} 7^2 b$ and $v_2(s') = v_2(s)$.
		By the inductive hypothesis, $s'$ is a winning state.
		By the definition of winning and losing states, $s$ is a losing state.
	\end{enumerate}	
\end{proof}

\begin{lemma} \label{lem:isWinStateAlg2}
	In Algorithm \ref{alg:isWinState}, the two formulas $\varphi_{i}$ and $\psi_i$ exactly captures the two sets of winning states with depth at most $i$ and of losing states with depth at most $i$, respectively.
\end{lemma}
\begin{proof}
	We prove by induction on $i$.
	
	\textbf{Base case} ($i = 0$):
	By the definition of winning and losing states, all terminal states are losing states and hence no winning state is a terminal state.
	By the definition of the depth on states, all terminal states are of depth $0$.
	Hence, $\varphi_0 = \false$ is the set of depth $0$ winning states, and $\psi_0 = \constraint \land \estate$ is the set of depth $0$ losing states.
	
	\textbf{Inductive step} ($i > 1$):
	By the inductive assumption, $\varphi_{i - 1}$ and $\psi_{i - 1}$ exactly captures the two sets of winning states with depth at most $i$ and of losing states with depth at most $i$, respectively.
		
	At Line 6, $\varphi_i = \exists \statevar' \{\trans(\act) \land (\constraint \land \psi_{i - 1})[\statevar/\statevar']\}$ denotes the set of states $s$ that has at least one successor state $s'$ s.t. $s'$ is losing state with at most depth $i - 1$.
	By the definition of depth on states, $\dep[s] \leq i$.
	By the definition of winning and losing states, $s$ is winning.
	Hence, $\varphi_i$ denotes the set of winning states with depth at most $i$.
	
	Let $s$ be a winning state with depth at most $i$.
	By the definition of winning and losing states, there is a successor state $s'$ of $s$ s.t. $s'$ is a losing state.
	By the definition of depth on states, $\dep[s'] \leq i - 1$.
	Hence, $s' \models \psi_{i - 1}$ and $s \models \exists \statevar' \{\trans(\act) \land (\constraint \land \psi_{i - 1})[\statevar/\statevar']\}$.
	
	Similarly, $\psi_i = \forall \statevar' \{\trans(\act) \rightarrow (\constraint \land \varphi_{i - 1})[\statevar/\statevar']\}$ is the set of losing states with depth at most $i$.
\end{proof}

\noindent
\textbf{Proof of Lemma \ref{lem:isWinStateAlg}}
\begin{proof}
Let $s$ be a state of a terminating and finite-depth ICG.
	
We first prove the termination of Algorithm \ref{alg:isWinState}.
In Alg. \ref{alg:isWinState}, it iteratively constructs the two class formulas $\varphi_{i}$'s and $\psi_i$'s.
By Lemma \ref{lem:isWinStateAlg2}, the two formulas $\varphi_{i}$ and $\psi_i$ exactly captures the two sets of winning states with depth at most $i$ and of losing states with depth at most $i$, respectively.
The $s$ must satisfy some $\psi_i$ or $\varphi_i$ (Lines 2, 8 \& 9).
Hence, the loop (Lines 4 - 10) is guaranteed to terminate.
Algorithm \ref{alg:isWinState} always terminates.

We now prove the soundness of Algorithm \ref{alg:isWinState} by contraposition.
Suppose that $s$ is not a winning state.
It follows that $s$ is a losing state.
It will satisfy some formula $\psi_i$ iteratively constructed in lines 1 \& 7.
Algorithm \ref{alg:isWinState} terminates with $\false$ (Lines 2 \& 9).

We finally prove the completeness of Algorithm \ref{alg:isWinState}.
Suppose that $s$ is a winning state.
It will satisfy some formula $\varphi_i$ iteratively constructed in line 1 \& 6	
Algorithm \ref{alg:isWinState} terminates with $\true$ (Line 8).	
\end{proof}

\section{Extensions to the Mis\`{e}re Play Rule} \label{sec:extensionToMisere}
\looseness=-1
In this section, we extend our decidability and undecidability results on ICGs under the normal play rule to the mis\`{e}re play rule.
The only difference between the normal and mis\`{e}re play rule lies in Definition \ref{def:wlState}.
Under the mis\`{e}re play rule, every terminal state is a winning state.
The undecidability and decidability results for LIA-definable ICGs summarized in Table \ref{tab:summary-results} also hold under the mis\`{e}re play rule.
In the following, we illustrate how to adjust the original proof under the normal play rule so as to suit the mis\`{e}re play rule rather than present the detailed proof for the mis\`{e}re play rule with the same idea as the original proof.

\subsection{The Terminating and Cyclic Problems}
The terminating and cyclic properties of ICGs do not depend on the play rule.
Thus, Theorems \ref{thm:isTerminatingICG}, \ref{thm:isCyclicICG}, and \ref{thm:isTerminatingCyclicICG} hold under the mis\`{e}re play rule.

\subsection{The Winning, Losing and Draw Formula Problems}
\looseness=-1
We first consider arbitrary ICGs.
Lemma \ref{lem:isWinForm} still holds except that the 4th item is 
\begin{enumerate}
	\item[4.] $s$ is a state of $\Pi'(M)$ s.t. $v(s) = 2^{l} 3^{m} 5^{n} 7^0 b$ (resp. $2^l 3^{m} 5^{n} 7^{1} b$), then $s$ is a losing (resp. winning) state iff there is a finite play from $s$.
\end{enumerate}

Hence, the formula $\phi_w: v \geq 0 \land \chi_{2,1,\len[L]}(v) \land \varphi_{7, 1}(v)$ exactly characterizes the set of positive integers $2^l 3^m 5^n 7^1 b$ where $1 \leq l \leq \len[L]$ and $b$ is not a multiple of $2$, $3$, $5$ or $7$.
It can be verified that $\phi_w$ is the winning formula for $\Pi'(M)$ iff $M$ is terminating on all configurations.
Similarly, $\phi_l: v \geq 0 \land \chi_{2,1,\len[L]}(v) \land \varphi_{7, 0}(v)$ is the losing formula iff $M$ is terminating on all configurations, and $\phi_d: \false$ is the draw formula iff $M$ is terminating on all configurations.
It directly follows that $\winFormProb$, $\loseFormProb$ and $\drawFormProb$ are undecidable for general ICGs.

\looseness=-1
In the following, we prove that $\winFormProb$, $\loseFormProb$ and $\drawFormProb$ become decidable for terminating ICGs.
To this end, we present the modified constraints for the winning formula under the mis\`{e}re play rule. 

\begin{definition}\label{def:consWinForm'} \rm
	Let $\Pi = \tuple{\var, \act, \constraint, \estate}$ be a terminating ICG.
	The constraints for the winning formula $\phi$ of $\Pi$ are as follows:
	\begin{enumerate}
		\item $\constraint \land \estate \rightarrow \phi$;
		\item $(\constraint \land \neg \estate \land \phi) \rightarrow \exists \statevar' [\trans(\act) \land (\constraint \land \neg \phi)[\statevar/\statevar']]$;
		\item $(\constraint \land \neg \phi) \rightarrow \forall \statevar' [\trans(\act) \rightarrow (\constraint \land \phi)[\statevar/\statevar']]$.
	\end{enumerate}
\end{definition}

\looseness=-1
The difference between Definitions \ref{def:consWinForm} and \ref{def:consWinForm'} are the following:
(1) the first constraint stipulates that every terminal state is winning rather than losing due to the difference between the normal and mis\`{e}re play rules; and
(2) the second constraint requires that every non-terminal winning state has some losing successors.
It can be verified that under the mis\`{e}re play rule, a formula $\phi$ is the winning formula of $\Pi$ iff all of the constraints illustrated in Definition \ref{def:consWinForm'} are valid.
It follows that $\winFormProb$ for terminating ICGs is decidable.
The set $\loseSet$ of losing states are the the complement of the $\winSet$ set of winning states relative to the set of legal states.
So $\loseFormProb$ for terminating ICGs is decidable.


\begin{algorithm}[!t]
	\small
	\caption{${\tt IsWinningState''}(\Pi, s)$}
	\label{alg:isWinState_Misere}
	\KwIn{$\Pi = \tuple{\var, \act, \constraint, \estate}$: a terminating finite-depth ICG; \\
		\hspace*{9.5mm}$s$: a legal state.}
	\KwOut{$\true$, if $s$ is a winning state under the mis\`{e}re play rule; $\false$, otherwise.}
	
	$\psi_{0} \assign \false$ and $\varphi_{0} \assign (\constraint \land \estate)$
	
	\lIf{$s \models \varphi_{0}$}
	{
		\Return{$\true$}
	}
	\Else{
		$i \assign 1$
		
		\While{{\bf true}}
		{
			$\psi_{i} \assign$ a quantifier-free formula equivalent to $\hspace*{6.5mm} \neg \estate \land \forall \statevar' \{\trans(\act) \rightarrow (\constraint \land \varphi_{i - 1})[\statevar/\statevar']\}$
			
			$\varphi_{i} \assign$ a quantifier-free formula equivalent to $\hspace*{6.5mm} \varphi_0 \lor \exists \statevar' \{\trans(\act) \land (\constraint \land \psi_{i - 1})[\statevar/\statevar']\}$
			
			\lIf{$s \models \constraint \land \psi_{i}$}
			{
				\Return{$\false$}
			}	
			
			\lIf{$s \models \constraint \land \varphi_{i}$}
			{
				\Return{$\true$}
			}			
			
			$i \assign i + 1$
		}
	}
\end{algorithm}

\subsection{The Winning, Losing and Draw State Problems}

\looseness=-1
We first consider the winning and losing state problems for terminating ICGs.
We define the forth reduction $\Pi'''(M)$ from 2-counter machines to ICGs that is the same as $\Pi''(M)$, except that the transition formula for $a_0$ is $\varphi_{7,0}(v_1) \land v'_1 = 7^2v_1 \land v'_2 \geq 0$.
Items 1, 2 and 3 of Lemma \ref{lem:isWinState} still hold for $\Pi'''(M)$ under the mis\`{e}re play rule.
Since terminal states are winning under the mis\`{e}re play rule, we present the modified Items 4 and 5.
\begin{enumerate}
    \item[4.] If $s$ is a state on a play from an initial state such that
    $v_1(s)=2^l3^m5^n7^1b$, then $s$ is winning (resp. losing) iff the play ends at a state $s^*$ such that
    $v_1(s^*)=2^{l^*}3^{m^*}5^{n^*}7^1b$
    (resp. $v_1(s^*)=2^{l^*}3^{m^*}5^{n^*}7^2b$).

    \item[5.] If $s$ is a state on a play from an initial state such that
    $v_1(s)=2^l3^m5^n7^2b$, then $s$ is losing (resp. winning) iff the play ends at a state $s^*$ such that
    $v_1(s^*)=2^{l^*}3^{m^*}5^{n^*}7^1b$
    (resp. $v_1(s^*)=2^{l^*}3^{m^*}5^{n^*}7^2b$).
\end{enumerate}

\looseness=-1
Let $C$ be a configuration of a 2-counter machine $M$ and $s$ a state of $\Pi'''(M)$ s.t. $v_1(s) = 2^{C(r_0)} 3^{C(r_1)} 5^{C(r_2)} 7^0 b$.
We can prove that the state $s$ is winning iff $M$ is terminating in $C$ under the mis\`{e}re play rule based on the above modified Lemma \ref{lem:isWinState} and the proof of Theorem \ref{thm:isWinState}.
%
Since the halting problem of 2-counter machines is undecidable, $\winStateProb$ remains undecidable even for terminating ICGs under the mis\`{e}re play rule.
In addition, the losing state problem is the complement of the winning state problem for terminating ICGs.
It directly follows that $\loseStateProb$ is also undecidable for terminating ICGs. 

\looseness=-1
We now consider the winning and losing state problems for terminating finite-depth ICGs.
To this end, we present the modified algorithm for the winning state problem and the modified definition of depth of winning or losing states under the mis\`{e}re play rule.
Similarly to Algorithm \ref{alg:isWinState}, Algorithm \ref{alg:isWinState_Misere} iteratively constructs two class formulas: $\varphi_i$'s and $\psi_i$'s, exactly capturing the two sets of winning states and of losing states with depth at most $i$ under the mis\`{e}re play rule.
The difference between Algorithms \ref{alg:isWinState} and \ref{alg:isWinState_Misere} are the following:
(1) Initially, to conform with the mis\`{e}re play rule, the formula $\varphi_0$ is $\constraint \land \estate$, meaning that terminal states are winning, and $\psi_0$ is $\false$, meaning that no terminal state is losing (Line 1); and 
(2) Each formula $\psi_i$ implies that $\neg \estate$, since $\forall \statevar' \{\trans(\act) \rightarrow (\constraint \land \varphi_{i - 1})[\statevar/\statevar']\}$ include terminal states which is not losing under the mis\`{e}re play rule; and 
(3) Each formula $\varphi_i$ contains a disjunct $\varphi_0$ since the set of winning states with depth at most $i$ contains all legal terminal states but $\exists \statevar' \{\trans(\act \land (\constraint \land \psi_{i - 1})[\statevar/\statevar']\}$ does not imply $\varphi_0$.
\begin{definition} \rm \label{def:depthMisere}
	Let $s$ be a winning or losing state of an ICG.
	The depth $d'(s)$ of $s$ under mis\`{e}re play rule is defined as:
		$\begin{cases}
			0 & \text{if } s \text{ is a terminal state}; \\
			
			\max \set{d'(s') \mid s' \in \suc(s)} + 1 & \text{if } s \text{ is a losing state}; \\
			
			\min \set{d'(s') \mid s' \in \loseSet \cap \suc(s)}  + 1 & \text{otherwise}. 
			%
		\end{cases}$
\end{definition}

\looseness=-1
Under the mis\`{e}re play rule, for terminating finite-depth ICGs, it can be verified that Algorithm \ref{alg:isWinState_Misere} enjoys the termination, soundness and completeness properties and hence $\winStateProb$ and $\loseStateProb$ are decidable.

\looseness=-1
We finally consider the draw state problem.
Item 1, 2, 3 and 5 of Lemma \ref{lem:isWinForm} still holds under the mis\`{e}re play rule.
It implies that Theorem \ref{thm:isDrawState} remains valid (that is, $\drawStateProb$ is undecidable) under the mis\`{e}re play rule.
In addition, as no draw state exists in any terminating ICG, the draw state problem for terminating ICGs is decidable.

\end{document}